\documentclass[runningheads]{llncs}

\usepackage[T1]{fontenc}
\usepackage[utf8]{inputenc}
\usepackage{amsmath,amssymb,mathtools}
\usepackage{graphicx}
\usepackage{tikz}
\definecolor{myblack}{RGB}{0,0,0}
\definecolor{border}{RGB}{206,206,206}
\DeclareRobustCommand{\sizeFtwo}{\Large}
\DeclareRobustCommand{\sizeFtwop}{\normalsize}
\usepackage{enumitem}
\usepackage[disable]{todonotes}
\usepackage{url}
\usepackage[most]{tcolorbox}
\usepackage{xcolor}
\newif\ifnamealgebra
\namealgebrafalse

\definecolor{thesissource}{RGB}{0,92,138}
\definecolor{ijucsource}{RGB}{0,78,110}
\definecolor{proposal}{RGB}{92,45,145}
\definecolor{changed}{RGB}{180,25,25}
\definecolor{todocolor}{RGB}{128,88,0}
\newtcolorbox{sourceblock}[2][]{breakable,enhanced,
  colback=#2!4,colframe=#2!75!black,boxrule=.45pt,
  left=4pt,right=4pt,top=3pt,bottom=3pt,
  title={Exact source: #1},fonttitle=\bfseries\footnotesize}
\newtcolorbox{proposalblock}[1][]{breakable,enhanced,
  colback=proposal!4,colframe=proposal!75!black,boxrule=.6pt,
  left=4pt,right=4pt,top=3pt,bottom=3pt,
  title={Proposed paper version: #1},fonttitle=\bfseries\footnotesize}
\newcommand{\changed}[1]{\textcolor{changed}{#1}}

\title{From Border-Agnosticism to Causal Rewrite Systems}
\author{Pablo Arrighi\inst{1} \and Marin Costes\inst{2} \and Luidnel Maignan\inst{3}}
\authorrunning{P. Arrighi, M. Costes, L. Maignan}
\institute{Universit\'e Paris-Saclay, Inria, CNRS, LMF, 91190 Gif-sur-Yvette, France\\
\email{pablo.arrighi@universite-paris-saclay.fr}
\and
Centre for Quantum Information and Communication, \'Ecole polytechnique de Bruxelles,
Universit\'e libre de Bruxelles, 1050 Brussels, Belgium\\
\email{marin.costes@ens-paris-saclay.fr}
\and
LACL, Univ Paris Est Creteil, 94000 Cr\'eteil, France\\
\email{luidnel.maignan@u-pec.fr}}

\newcommand{\calC}{\mathcal{C}}
\newcommand{\calD}{\mathcal{D}}
\newcommand{\calG}{\mathcal{G}}
\newcommand{\calR}{\mathcal{R}}
\newcommand{\calS}{\mathcal{S}}
\newcommand{\calV}{\mathcal{V}}
\newcommand{\calX}{\mathcal{X}}
\newcommand{\I}{\mathrm{I}}
\newcommand{\B}{\mathrm{B}}
\newcommand{\E}{\mathrm{E}}
\newcommand{\V}{\mathrm{V}}
\newcommand{\Past}{\mathrm{Past}}
\newcommand{\X}{\calX}
\newcommand{\sshift}{\calS}
\newcommand{\D}{\calD}
\renewcommand{\sp}[1]{\lfloor #1\rfloor}
\newcommand{\namealg}[1]{\widehat{#1}}
\newcommand{\nameeq}{\,\hat{=}\,}
\newcommand{\namesubseteq}{\,\hat{\subseteq}\,}
\newcommand{\namecap}{\,\hat{\cap}\,}
\newcommand{\possubseteq}[2]{\ifnamealgebra #1\namesubseteq #2\else
  \sp{#1}\subseteq\sp{#2}\fi}
\newcommand{\posdisjoint}[2]{\ifnamealgebra \neg(#1\namecap #2)\else
  \sp{#1}\cap\sp{#2}=\varnothing\fi}
\newcommand{\restr}[2]{\mathcal{N}_{#1}(#2)}
\newcommand{\restri}[1]{\mathcal{N}_{#1}}
\newcommand{\crestri}[1]{\overline{\mathcal{N}}_{#1}}
\newcommand{\restriA}[1]{\mathcal{N}_{#1}}
\newcommand{\restriD}[2]{\mathcal{N}^{#1}_{#2}}
\newcommand{\crestriD}[2]{\overline{\mathcal{N}}^{#1}_{#2}}
\newcommand{\restrD}[3]{\mathcal{N}^{#1}_{#2}(#3)}
\newcommand{\graphvalidexplicit}[2]{\Gamma_{#2}}
\newcommand{\graphvalid}[2]{\graphvalidexplicit{#1}{#2}}
\newcommand{\graphvalidmutex}[2]{\Gamma^{#1}_{#2}}
\newcommand{\diskA}[1]{\sshift_{\restriA{#1}}}
\newcommand{\dangling}[1]{\E_{#1}\setminus(\I_{#1}\!:\!\pi)^2}
\newcommand{\xCone}{\calC_x}
\newcommand{\xcone}[1]{\calC_{#1}}

\newcommand{\calP}{\mathcal{P}}
\newcommand{\position}{position}
\definecolor{rxone}{RGB}{0,110,160}
\definecolor{rxtwo}{RGB}{170,85,0}
\newcommand{\RX}[2]{%
  #1
}

\newcommand{\restriB}[1]{\mathcal{B}_{#1}}
\newcommand{\diskB}[1]{\sshift_{\restriB{#1}}}

\begin{document}
\maketitle

\begin{abstract}
Causal Graph Rewriting was recently introduced as an evolution of Causal Graph Dynamics to allow for a controlled form of asynchronism in these dynamical systems with dynamical graphs.
In this context, the notion of borders---akin to dangling edges---has gained increasing importance.
In this paper, we re-examine these borders, interpreting them as missing information.
This perspective adds structural rigor to the framework, allowing border-agnostic causal graph rewriting to be reformulated as a rewriting system rather than a dynamical system.
This shift represents a crucial step toward facilitating the treatment of reversibility.
\end{abstract}

\section{Introduction}\label{sec:introduction}

In a rewrite system, locality is visible in the syntax: a left-hand side is
matched and replaced while the context is retained . In
the dynamical-systems tradition, on the other hand, locality is often
axiomatic. For instance, cellular automata  are global maps \emph{induced} by
a uniform local map; causal graph dynamics extend this viewpoint to spaces
whose graph structure may itself
evolve~\cite{ArrighiCGD,ArrighiRCGD}. This axiomatic presentation is also
classical in their quantum extensions~\cite{ArrighiQCGD,ArrighiQNT}.

Causal graph rewriting, introduced recently, stands between these two
traditions~\cite{ArrighiSpaceTimeDetWRLA,arrighi2026causalgraphrewriting}.
Updates are asynchronous, and directed edges encode both connectivity and
dependencies between events. Vertices carry time tags, and a local rule
consumes a ready event and rewrites a neighbourhood. Graphs may be partial
views, represented by border edges and border vertices. The link with
rewrite theory is the asynchronous regime itself, in which each update may
be scheduled freely. The definition of local rules, however, borrows from
the axiomatic side: a local rule acting at \(x\) on a graph \(G\) is a
\emph{global} function of the form \(A_xG=A_x(G_{\restri{x}}) \sqcup G_{\crestri{x}}\),
where \(G_{\restri{x}}\) is the neighbourhood of \(x\) in \(G\) and
\(G_{\crestri{x}}\) the rest of the graph. The neighbourhood scheme
\(\restri{}\) is only required to be
\emph{extensive}~\cite{ArrighiSpaceTimeDetWRLA,arrighi2026causalgraphrewriting}
(called extensibility in~\cite{ArrighiQNT}): restricting a graph to an
intermediate view that still contains the selected neighbourhood should produce the same neighborhood. This condition implies idempotence (the neighborhood of a neighborhood is that neighborhood), which is needed to pass to
the quantum setting, and it sufficed for the results obtained so far.

\paragraph{Taking borders seriously.}
The main point of this paper is that being careful with the meaning of
borders provides a first solid ground to unify causal graph rewriting with
traditional graph rewriting and rewrite systems. A border vertex should read as missing information, i.e.\ a part
of the graph that is not known, rather than the end of the graph, in the same spirit as~\cite{DBLP:conf/gg/MaignanS24}.
Extensivity does not respect this reading: it only compares a graph with
its subgraphs, so a neighbourhood selected on a partial view may still
change when information is revealed beyond its borders. Taking the missing-information reading
seriously, we restrict to \emph{border-agnostic} neighbourhood schemes, which
select the same neighbourhood whatever is revealed around it. This
restriction is natural: it is equivalent to a safety principle, stating that
a local rule cannot act with less information than a neighbourhood. It also yields the two properties needed to
present local rules as rewrite systems. First, \emph{disk uniqueness}: a
graph contains at most one disk at each position, so that matching a
left-hand side is unambiguous. Second, \emph{disk-completability}: every
disk is the neighbourhood of its position in a graph where all its
vertices are internal, so that the structural properties
of~\cite{arrighi2026causalgraphrewriting}, and in particular the preservation
of border edges, hold on every disk. As in graph rewriting, where border
edges must be preserved by pattern substitution for the global
transformation to be well defined, this allows a uniform description of the
action of local rules on disks. Building on these two properties, we show
that local rules are exactly the actions of \emph{causal rewrite systems},
families of pattern substitutions subject to conditions that can mostly be
checked rule by rule.

\paragraph{Contributions.}
\begin{enumerate}[leftmargin=*,itemsep=.25ex]
\ifnamealgebra
\item We record the name algebra used to split and merge vertices without
choosing globally fresh names. It adapts~\cite{ArrighiCreation} so that time
tags commute with these operations, and so that the later structural lemmas
can speak of position fragments.
\fi
\item We introduce border-agnosticism (Def.~\ref{def:strong-ext}), separate
it from extensivity (Ex.~\ref{ex:rounded}), and show that it is equivalent
to the safety principle (Prop.~\ref{prop:safety}).
\item We show that border-agnosticism is equivalent to disk uniqueness
(Lem.~\ref{rem:disk-unique}) and, on border-completable sets of graphs,
implies disk-completability (Lem.~\ref{lem:disk-completable},
Ex.~\ref{ex:neighbours}). The latter removes the extendability premise from
the structural properties of local rules
of~\cite{arrighi2026causalgraphrewriting}.\ifnamealgebra{} These properties
also hold for the name algebra (Lem.~\ref{lem:pos-preservation}).\fi
\item Our main result is a representation theorem: border-agnostic
neighbourhood schemes correspond exactly to mutex sets of cones
(Prop.~\ref{prop:neigh-mutex}), and local rules correspond exactly to causal
rewrite systems (Thm.~\ref{thm:representation}).
\end{enumerate}

Sec.~\ref{sec:setting} recalls causal graph rewriting.
Sec.~\ref{sec:strong} introduces border-agnosticism and its consequences.
Sec.~\ref{sec:crs} proves the representation theorem, and
Sec.~\ref{sec:related} discusses related work.

\section{Causal graphs and local rules}\label{sec:setting}
We briefly recall the formalism of causal graph rewriting; a gentler
introduction can be found in~\cite{arrighi2026causalgraphrewriting}.
\ifnamealgebra
\else

\subsection{Graphs and locality}\label{sec:graphs}

We work with directed acyclic labelled port graphs: edges are attached to
ports of the vertices, and each internal vertex carries a state. Labelled
port graphs are standard in distributed
computing~\cite{PapazianRemila,Chalopin,KrivineKappa}, where ports
distinguish neighbouring processes and labels represent their states;
directed acyclic graphs commonly capture dependencies between
processes~\cite{CausalSets,AsynchronousSimsSync}. A vertex \(t.x\) is an
event: a process at position \(x\) with time tag \(t\). Graphs may be partial
views, with \emph{border} edges linking internal vertices to border vertices,
whose states are unknown.

\begin{definition}[Positions, names, ports, states]\label{def:namealgebra}
Fix a countably infinite set $\calX$ of \emph{\position s}.
Let us denote $\calV := \{\, t.x \mid (t,x) \in \mathbb{Z} \times \calX \,\}$ and call its elements \emph{names}. We say that $u = t.x\in \calV$ is at position $x$ and denote $t'.u := (t'+t).x$ for any $t' \in \mathbb{Z}$. For any subset $U \subseteq \calV$, we denote $\sp{U} := \{\, x\in \calX \mid \exists t\in \mathbb{Z}, t.x\in U \,\}$ the set of underlying positions.
 For any $X\subseteq \mathcal{X}$, let us denote $\overline{X} := \calX \setminus X$.
Finally, fix a finite sets of \emph{ports} \(\pi\) and \emph{states} \(\Sigma\).
\end{definition}
\fi

\begin{definition}[Graphs]\label{def:graphs}
A graph $G$ is given by a tuple $(\I_G,\B_G,\E_G,\sigma_G)$ where:
\begin{itemize}[leftmargin=*]
\item $\I_G \subseteq \calV$ is the set of \emph{internal vertices of $G$},
\item $\B_G \subseteq \calV \setminus \I_G$ is the set of \emph{border
vertices of $G$},
\item $\E_G \subseteq (\V_G\!:\!\pi )^{2} \setminus
( \B_{G}\!:\!\pi )^2$ is the set of \emph{(oriented) edges}, and
\item $\sigma_G :\I_G\rightarrow \Sigma$ maps each internal vertex to its
state.
\end{itemize}
where we denote by $\V_G := \I_G\cup \B_G$ the set of all vertices of the
graph, and by $(V\!:\!\pi):=\{v\!:\!p\mid v\in V,\,p\in\pi\}$ the set of ports
of some set of vertices $V$. Moreover the tuple has to be such that:
\begin{description}[leftmargin=!,labelwidth=3.4cm,itemsep=.25ex]
\item[acyclicity:] the graph has no cycles.
\item[border-attachment:]
$\forall u \in \B_G,\, \exists (v\!:\!a,v'\!:\!a') \in \E_G,\,
u \in \{\,v,v'\,\}$, i.e.\ border vertices are connected to at least one
edge.
\item[port-saturation:]
$\forall (u\!:\!a,v\!:\!b),(u'\!:\!a',v'\!:\!b') \in {\E_G},\,
u\!:\!a\neq v'\!:\!b' \,\wedge\,
(u\!:\!a=u'\!:\!a' \Leftrightarrow v\!:\!b=v'\!:\!b')$, i.e.\ ports are used
only once---so as to distinguish each neighbour.
\ifnamealgebra
\item[non-overlapping positions:]%
\[
\begin{gathered}
\changed{\forall u,v\in\V_G,\ \forall p,p'\in\{l,r\}^*,}\\
\changed{\sp{u}.p=\sp{v}.p'\ \Longrightarrow\ u=v.}
\end{gathered}
\]
i.e.\ position fragments appear only once---so as to avoid name conflicts
when merging or splitting names.
\else
\item[unicity of positions:]
$\forall u,v\in\V_G,\ \sp{u}=\sp{v}\Longrightarrow u=v$, i.e.\ each
position is carried by at most one vertex.
\fi
\item[finite chains:] the graph only contains paths of finite length.
\end{description}
We denote by $\Past(G)\subseteq\I_G$ the vertices with no incoming edges and
by $\calG$ the set of all graphs. We call \emph{border edges} the elements of \(\dangling{G}\).
\end{definition}

The local evolution of these graphs is governed by three main operations.

First, for \(U\subseteq\calV\), the induced subgraph \(G_U\) has internal vertices
\(\I_G\cap U\) and all edges touching one of them; adjacent vertices outside
\(U\) become border vertices (Fig.~\ref{fig:induced-subgraph}). We write \(H\sqsubseteq G\) when \(H\) is an
induced subgraph of \(G\). For \(X\subseteq\calX\), \(G_X\) denotes induction
on names at positions in \(X\).

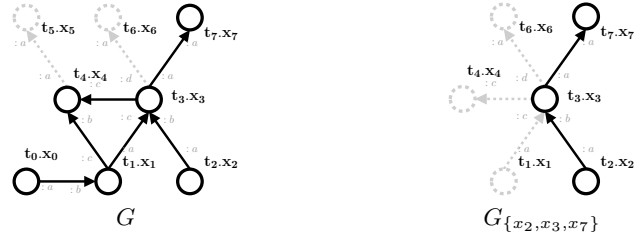
\begin{figure}[t]
\hfil
\begin{minipage}[b]{0.253\textwidth}
    \centering
    \resizebox{\textwidth}{!}{\tikzset{every picture/.style={line width=0.75pt}} %set default line width to 0.75pt        

\begin{tikzpicture}[x=0.75pt,y=0.75pt,yscale=-1,xscale=1]
%uncomment if require: \path (0,312); %set diagram left start at 0, and has height of 312

%Straight Lines [id:da006076172788554324] 
\draw [line width=2.25]    (220,185) -- (155,185) ;
\draw [shift={(150,185)}, rotate = 360] [fill=myblack  ][line width=0.08]  [draw opacity=0] (14.29,-6.86) -- (0,0) -- (14.29,6.86) -- cycle    ;
%Straight Lines [id:da9427679372526456] 
\draw [line width=2.25]    (100,285) -- (165,285) ;
\draw [shift={(170,285)}, rotate = 180] [fill=myblack  ][line width=0.08]  [draw opacity=0] (14.29,-6.86) -- (0,0) -- (14.29,6.86) -- cycle    ;
%Straight Lines [id:da3053547048924985] 
\draw [color=border  ,draw opacity=1 ][line width=2.25]  [dash pattern={on 2.53pt off 3.02pt}]  (135,170) -- (87.91,104.07) ;
\draw [shift={(85,100)}, rotate = 54.46] [fill=border  ,fill opacity=1 ][line width=0.08]  [draw opacity=0] (14.29,-6.86) -- (0,0) -- (14.29,6.86) -- cycle    ;
%Straight Lines [id:da07898649264243718] 
\draw [line width=2.25]    (185,270) -- (137.91,204.07) ;
\draw [shift={(135,200)}, rotate = 54.46] [fill=myblack  ][line width=0.08]  [draw opacity=0] (14.29,-6.86) -- (0,0) -- (14.29,6.86) -- cycle    ;
%Straight Lines [id:da6885992751563419] 
\draw [color=border  ,draw opacity=1 ][line width=2.25]  [dash pattern={on 2.53pt off 3.02pt}]  (235,170) -- (187.91,104.07) ;
\draw [shift={(185,100)}, rotate = 54.46] [fill=border  ,fill opacity=1 ][line width=0.08]  [draw opacity=0] (14.29,-6.86) -- (0,0) -- (14.29,6.86) -- cycle    ;
%Straight Lines [id:da37379453049535494] 
\draw [line width=2.25]    (285,270) -- (237.91,204.07) ;
\draw [shift={(235,200)}, rotate = 54.46] [fill=myblack  ][line width=0.08]  [draw opacity=0] (14.29,-6.86) -- (0,0) -- (14.29,6.86) -- cycle    ;
%Straight Lines [id:da5483691007548878] 
\draw [line width=2.25]    (185,270) -- (232.09,204.07) ;
\draw [shift={(235,200)}, rotate = 125.54] [fill=myblack  ][line width=0.08]  [draw opacity=0] (14.29,-6.86) -- (0,0) -- (14.29,6.86) -- cycle    ;
%Straight Lines [id:da38967403263107525] 
\draw [line width=2.25]    (235,170) -- (282.09,104.07) ;
\draw [shift={(285,100)}, rotate = 125.54] [fill=myblack  ][line width=0.08]  [draw opacity=0] (14.29,-6.86) -- (0,0) -- (14.29,6.86) -- cycle    ;
%Shape: Circle [id:dp2683748824167219] 
\draw  [color=myblack  ,draw opacity=1 ][fill={rgb, 255:red, 255; green, 255; blue, 255 }  ,fill opacity=1 ][line width=3]  (170,285) .. controls (170,276.72) and (176.72,270) .. (185,270) .. controls (193.28,270) and (200,276.72) .. (200,285) .. controls (200,293.28) and (193.28,300) .. (185,300) .. controls (176.72,300) and (170,293.28) .. (170,285) -- cycle ;
%Shape: Circle [id:dp5756000109789217] 
\draw  [color=myblack  ,draw opacity=1 ][fill={rgb, 255:red, 255; green, 255; blue, 255 }  ,fill opacity=1 ][line width=3]  (70,285) .. controls (70,276.72) and (76.72,270) .. (85,270) .. controls (93.28,270) and (100,276.72) .. (100,285) .. controls (100,293.28) and (93.28,300) .. (85,300) .. controls (76.72,300) and (70,293.28) .. (70,285) -- cycle ;
%Shape: Circle [id:dp11221956174907122] 
\draw  [color=myblack  ,draw opacity=1 ][fill={rgb, 255:red, 255; green, 255; blue, 255 }  ,fill opacity=1 ][line width=3]  (270,285) .. controls (270,276.72) and (276.72,270) .. (285,270) .. controls (293.28,270) and (300,276.72) .. (300,285) .. controls (300,293.28) and (293.28,300) .. (285,300) .. controls (276.72,300) and (270,293.28) .. (270,285) -- cycle ;
%Shape: Circle [id:dp4345523891066313] 
\draw  [color=myblack  ,draw opacity=1 ][fill={rgb, 255:red, 255; green, 255; blue, 255 }  ,fill opacity=1 ][line width=3]  (120,185) .. controls (120,176.72) and (126.72,170) .. (135,170) .. controls (143.28,170) and (150,176.72) .. (150,185) .. controls (150,193.28) and (143.28,200) .. (135,200) .. controls (126.72,200) and (120,193.28) .. (120,185) -- cycle ;
%Shape: Circle [id:dp1573965799231728] 
\draw  [color=myblack  ,draw opacity=1 ][fill={rgb, 255:red, 255; green, 255; blue, 255 }  ,fill opacity=1 ][line width=3]  (220,185) .. controls (220,176.72) and (226.72,170) .. (235,170) .. controls (243.28,170) and (250,176.72) .. (250,185) .. controls (250,193.28) and (243.28,200) .. (235,200) .. controls (226.72,200) and (220,193.28) .. (220,185) -- cycle ;
%Shape: Circle [id:dp39702060276476947] 
\draw  [color=border  ,draw opacity=1 ][fill={rgb, 255:red, 255; green, 255; blue, 255 }  ,fill opacity=1 ][dash pattern={on 3.38pt off 3.27pt}][line width=3]  (170,85) .. controls (170,76.72) and (176.72,70) .. (185,70) .. controls (193.28,70) and (200,76.72) .. (200,85) .. controls (200,93.28) and (193.28,100) .. (185,100) .. controls (176.72,100) and (170,93.28) .. (170,85) -- cycle ;
%Shape: Circle [id:dp01297082093593771] 
\draw  [color=border  ,draw opacity=1 ][fill={rgb, 255:red, 255; green, 255; blue, 255 }  ,fill opacity=1 ][dash pattern={on 3.38pt off 3.27pt}][line width=3]  (70,85) .. controls (70,76.72) and (76.72,70) .. (85,70) .. controls (93.28,70) and (100,76.72) .. (100,85) .. controls (100,93.28) and (93.28,100) .. (85,100) .. controls (76.72,100) and (70,93.28) .. (70,85) -- cycle ;
%Shape: Circle [id:dp4122806456108977] 
\draw  [color=myblack  ,draw opacity=1 ][fill={rgb, 255:red, 255; green, 255; blue, 255 }  ,fill opacity=1 ][line width=3]  (270,85) .. controls (270,76.72) and (276.72,70) .. (285,70) .. controls (293.28,70) and (300,76.72) .. (300,85) .. controls (300,93.28) and (293.28,100) .. (285,100) .. controls (276.72,100) and (270,93.28) .. (270,85) -- cycle ;

% Text Node
\draw (81,242.4) node [anchor=north west][inner sep=0.75pt]  [font=\sizeFtwo]  {$\mathbf{t_{0} .x_{0}}$};
% Text Node
\draw (201,252.4) node [anchor=north west][inner sep=0.75pt]  [font=\sizeFtwo]  {$\mathbf{t_{1} .x_{1}}$};
% Text Node
\draw (301,252.4) node [anchor=north west][inner sep=0.75pt]  [font=\sizeFtwo]  {$\mathbf{t_{2} .x_{2}}$};
% Text Node
\draw (141,145.4) node [anchor=north west][inner sep=0.75pt]  [font=\sizeFtwo]  {$\mathbf{t_{4} .x_{4}}$};
% Text Node
\draw (261,172.4) node [anchor=north west][inner sep=0.75pt]  [font=\sizeFtwo]  {$\mathbf{t_{3} .x_{3}}$};
% Text Node
\draw (102,88.4) node [anchor=north west][inner sep=0.75pt]  [font=\sizeFtwo]  {$\mathbf{t_{5} .x_{5}}$};
% Text Node
\draw (202,88.4) node [anchor=north west][inner sep=0.75pt]  [font=\sizeFtwo]  {$\mathbf{t_{6} .x_{6}}$};
% Text Node
\draw (301,95.4) node [anchor=north west][inner sep=0.75pt]  [font=\sizeFtwo]  {$\mathbf{t_{7} .x_{7}}$};
% Text Node
\draw (279,242.4) node [anchor=north west][inner sep=0.75pt]  [font=\sizeFtwop,color={rgb, 255:red, 155; green, 155; blue, 155 }  ,opacity=1 ]  {$:a$};
% Text Node
\draw (201,242.4) node [anchor=north west][inner sep=0.75pt]  [font=\sizeFtwop,color={rgb, 255:red, 155; green, 155; blue, 155 }  ,opacity=1 ]  {$:a$};
% Text Node
\draw (102,288.4) node [anchor=north west][inner sep=0.75pt]  [font=\sizeFtwop,color={rgb, 255:red, 155; green, 155; blue, 155 }  ,opacity=1 ]  {$:a$};
% Text Node
\draw (99,152.4) node [anchor=north west][inner sep=0.75pt]  [font=\sizeFtwop,color={rgb, 255:red, 155; green, 155; blue, 155 }  ,opacity=1 ]  {$:a$};
% Text Node
\draw (69,112.4) node [anchor=north west][inner sep=0.75pt]  [font=\sizeFtwop,color={rgb, 255:red, 155; green, 155; blue, 155 }  ,opacity=1 ]  {$:a$};
% Text Node
\draw (169,112.4) node [anchor=north west][inner sep=0.75pt]  [font=\sizeFtwop,color={rgb, 255:red, 155; green, 155; blue, 155 }  ,opacity=1 ]  {$:a$};
% Text Node
\draw (279,112.4) node [anchor=north west][inner sep=0.75pt]  [font=\sizeFtwop,color={rgb, 255:red, 155; green, 155; blue, 155 }  ,opacity=1 ]  {$:a$};
% Text Node
\draw (249,152.4) node [anchor=north west][inner sep=0.75pt]  [font=\sizeFtwop,color={rgb, 255:red, 155; green, 155; blue, 155 }  ,opacity=1 ]  {$:a$};
% Text Node
\draw (251,202.4) node [anchor=north west][inner sep=0.75pt]  [font=\sizeFtwop,color={rgb, 255:red, 155; green, 155; blue, 155 }  ,opacity=1 ]  {$:b$};
% Text Node
\draw (199,202.4) node [anchor=north west][inner sep=0.75pt]  [font=\sizeFtwop,color={rgb, 255:red, 155; green, 155; blue, 155 }  ,opacity=1 ]  {$:c$};
% Text Node
\draw (137,288.4) node [anchor=north west][inner sep=0.75pt]  [font=\sizeFtwop,color={rgb, 255:red, 155; green, 155; blue, 155 }  ,opacity=1 ]  {$:b$};
% Text Node
\draw (151,202.4) node [anchor=north west][inner sep=0.75pt]  [font=\sizeFtwop,color={rgb, 255:red, 155; green, 155; blue, 155 }  ,opacity=1 ]  {$:b$};
% Text Node
\draw (161,162.4) node [anchor=north west][inner sep=0.75pt]  [font=\sizeFtwop,color={rgb, 255:red, 155; green, 155; blue, 155 }  ,opacity=1 ]  {$:c$};
% Text Node
\draw (198,152.4) node [anchor=north west][inner sep=0.75pt]  [font=\sizeFtwop,color={rgb, 255:red, 155; green, 155; blue, 155 }  ,opacity=1 ]  {$:d$};
% Text Node
\draw (151,252.4) node [anchor=north west][inner sep=0.75pt]  [font=\sizeFtwop,color={rgb, 255:red, 155; green, 155; blue, 155 }  ,opacity=1 ]  {$:c$};

\end{tikzpicture}}\\
    \(G\)
\end{minipage}
\hfil
\begin{minipage}[b]{0.42\textwidth}
    \centering
    \resizebox{.5\textwidth}{!}{\tikzset{every picture/.style={line width=0.75pt}} %set default line width to 0.75pt        

\begin{tikzpicture}[x=0.75pt,y=0.75pt,yscale=-1,xscale=1]
%uncomment if require: \path (0,312); %set diagram left start at 0, and has height of 312

%Straight Lines [id:da25778108431329194] 
\draw [color=border  ,draw opacity=1 ][line width=2.25]  [dash pattern={on 2.53pt off 3.02pt}]  (222,185) -- (157,185) ;
\draw [shift={(152,185)}, rotate = 360] [fill=border  ,fill opacity=1 ][line width=0.08]  [draw opacity=0] (14.29,-6.86) -- (0,0) -- (14.29,6.86) -- cycle    ;
%Straight Lines [id:da5330872910763639] 
\draw [color=border  ,draw opacity=1 ][line width=2.25]  [dash pattern={on 2.53pt off 3.02pt}]  (237,170) -- (189.91,104.07) ;
\draw [shift={(187,100)}, rotate = 54.46] [fill=border  ,fill opacity=1 ][line width=0.08]  [draw opacity=0] (14.29,-6.86) -- (0,0) -- (14.29,6.86) -- cycle    ;
%Straight Lines [id:da739612951368901] 
\draw [line width=2.25]    (287,270) -- (239.91,204.07) ;
\draw [shift={(237,200)}, rotate = 54.46] [fill=myblack  ][line width=0.08]  [draw opacity=0] (14.29,-6.86) -- (0,0) -- (14.29,6.86) -- cycle    ;
%Straight Lines [id:da39963266218894833] 
\draw [color=border  ,draw opacity=1 ][line width=2.25]  [dash pattern={on 2.53pt off 3.02pt}]  (187,270) -- (234.09,204.07) ;
\draw [shift={(237,200)}, rotate = 125.54] [fill=border  ,fill opacity=1 ][line width=0.08]  [draw opacity=0] (14.29,-6.86) -- (0,0) -- (14.29,6.86) -- cycle    ;
%Straight Lines [id:da5285703650861084] 
\draw [line width=2.25]    (237,170) -- (284.09,104.07) ;
\draw [shift={(287,100)}, rotate = 125.54] [fill=myblack  ][line width=0.08]  [draw opacity=0] (14.29,-6.86) -- (0,0) -- (14.29,6.86) -- cycle    ;
%Shape: Circle [id:dp8817091719855964] 
\draw  [color=border  ,draw opacity=1 ][fill={rgb, 255:red, 255; green, 255; blue, 255 }  ,fill opacity=1 ][dash pattern={on 3.38pt off 3.27pt}][line width=3]  (172,285) .. controls (172,276.72) and (178.72,270) .. (187,270) .. controls (195.28,270) and (202,276.72) .. (202,285) .. controls (202,293.28) and (195.28,300) .. (187,300) .. controls (178.72,300) and (172,293.28) .. (172,285) -- cycle ;
%Shape: Circle [id:dp7569884782316096] 
\draw  [color=myblack  ,draw opacity=1 ][fill={rgb, 255:red, 255; green, 255; blue, 255 }  ,fill opacity=1 ][line width=3]  (272,285) .. controls (272,276.72) and (278.72,270) .. (287,270) .. controls (295.28,270) and (302,276.72) .. (302,285) .. controls (302,293.28) and (295.28,300) .. (287,300) .. controls (278.72,300) and (272,293.28) .. (272,285) -- cycle ;
%Shape: Circle [id:dp24488145344431256] 
\draw  [color=border  ,draw opacity=1 ][fill={rgb, 255:red, 255; green, 255; blue, 255 }  ,fill opacity=1 ][dash pattern={on 3.38pt off 3.27pt}][line width=3]  (122,185) .. controls (122,176.72) and (128.72,170) .. (137,170) .. controls (145.28,170) and (152,176.72) .. (152,185) .. controls (152,193.28) and (145.28,200) .. (137,200) .. controls (128.72,200) and (122,193.28) .. (122,185) -- cycle ;
%Shape: Circle [id:dp8886770338580505] 
\draw  [color=myblack  ,draw opacity=1 ][fill={rgb, 255:red, 255; green, 255; blue, 255 }  ,fill opacity=1 ][line width=3]  (222,185) .. controls (222,176.72) and (228.72,170) .. (237,170) .. controls (245.28,170) and (252,176.72) .. (252,185) .. controls (252,193.28) and (245.28,200) .. (237,200) .. controls (228.72,200) and (222,193.28) .. (222,185) -- cycle ;
%Shape: Circle [id:dp15265138264901146] 
\draw  [color=border  ,draw opacity=1 ][fill={rgb, 255:red, 255; green, 255; blue, 255 }  ,fill opacity=1 ][dash pattern={on 3.38pt off 3.27pt}][line width=3]  (172,85) .. controls (172,76.72) and (178.72,70) .. (187,70) .. controls (195.28,70) and (202,76.72) .. (202,85) .. controls (202,93.28) and (195.28,100) .. (187,100) .. controls (178.72,100) and (172,93.28) .. (172,85) -- cycle ;
%Shape: Circle [id:dp5004182012597886] 
\draw  [color=myblack  ,draw opacity=1 ][fill={rgb, 255:red, 255; green, 255; blue, 255 }  ,fill opacity=1 ][line width=3]  (272,85) .. controls (272,76.72) and (278.72,70) .. (287,70) .. controls (295.28,70) and (302,76.72) .. (302,85) .. controls (302,93.28) and (295.28,100) .. (287,100) .. controls (278.72,100) and (272,93.28) .. (272,85) -- cycle ;

% Text Node
\draw (203,252.4) node [anchor=north west][inner sep=0.75pt]  [font=\sizeFtwo]  {$\mathbf{t_{1} .x_{1}}$};
% Text Node
\draw (303,252.4) node [anchor=north west][inner sep=0.75pt]  [font=\sizeFtwo]  {$\mathbf{t_{2} .x_{2}}$};
% Text Node
\draw (141,142.4) node [anchor=north west][inner sep=0.75pt]  [font=\sizeFtwo]  {$\mathbf{t_{4} .x_{4}}$};
% Text Node
\draw (263,172.4) node [anchor=north west][inner sep=0.75pt]  [font=\sizeFtwo]  {$\mathbf{t_{3} .x_{3}}$};
% Text Node
\draw (204,88.4) node [anchor=north west][inner sep=0.75pt]  [font=\sizeFtwo]  {$\mathbf{t_{6} .x_{6}}$};
% Text Node
\draw (303,95.4) node [anchor=north west][inner sep=0.75pt]  [font=\sizeFtwo]  {$\mathbf{t_{7} .x_{7}}$};
% Text Node
\draw (281,242.4) node [anchor=north west][inner sep=0.75pt]  [font=\sizeFtwop,color={rgb, 255:red, 155; green, 155; blue, 155 }  ,opacity=1 ]  {$:a$};
% Text Node
\draw (203,242.4) node [anchor=north west][inner sep=0.75pt]  [font=\sizeFtwop,color={rgb, 255:red, 155; green, 155; blue, 155 }  ,opacity=1 ]  {$:a$};
% Text Node
\draw (171,112.4) node [anchor=north west][inner sep=0.75pt]  [font=\sizeFtwop,color={rgb, 255:red, 155; green, 155; blue, 155 }  ,opacity=1 ]  {$:a$};
% Text Node
\draw (281,112.4) node [anchor=north west][inner sep=0.75pt]  [font=\sizeFtwop,color={rgb, 255:red, 155; green, 155; blue, 155 }  ,opacity=1 ]  {$:a$};
% Text Node
\draw (251,152.4) node [anchor=north west][inner sep=0.75pt]  [font=\sizeFtwop,color={rgb, 255:red, 155; green, 155; blue, 155 }  ,opacity=1 ]  {$:a$};
% Text Node
\draw (253,202.4) node [anchor=north west][inner sep=0.75pt]  [font=\sizeFtwop,color={rgb, 255:red, 155; green, 155; blue, 155 }  ,opacity=1 ]  {$:b$};
% Text Node
\draw (201,202.4) node [anchor=north west][inner sep=0.75pt]  [font=\sizeFtwop,color={rgb, 255:red, 155; green, 155; blue, 155 }  ,opacity=1 ]  {$:c$};
% Text Node
\draw (163,162.4) node [anchor=north west][inner sep=0.75pt]  [font=\sizeFtwop,color={rgb, 255:red, 155; green, 155; blue, 155 }  ,opacity=1 ]  {$:c$};
% Text Node
\draw (200,152.4) node [anchor=north west][inner sep=0.75pt]  [font=\sizeFtwop,color={rgb, 255:red, 155; green, 155; blue, 155 }  ,opacity=1 ]  {$:d$};

\end{tikzpicture}}\\
    \(G_{\{x_2,x_3,x_7\}}\)
\end{minipage}
\caption{A graph $G$ (left) and its induced subgraph
$G_{\{x_2,x_3,x_7\}}$ (right). Border vertices and border edges are
dashed. Note how internal vertices of $G_{\{x_2,x_3,x_7\}}$ are exactly $x_2$, $x_3$, and $x_7$. Its border vertices are vertices of the original graph (internal and border) in the immediate vaccinity of $x_2$ ,$x_3$, and $x_7$.}
\label{fig:induced-subgraph}
\end{figure}

Second, the partial composition \(G\sqcup H\) is
defined when both are induced subgraphs of a common graph, and is then their
least common supergraph; in particular \(G=G_X\sqcup G_{\overline X}\) for
every \(X\subseteq\calX\).
\ifnamealgebra
With the name algebra, the position clause of that lemma is the
non-overlapping condition above: sharing a position fragment forces the two
vertices to coincide.
\else
\fi
Applications usually restrict to a family of graphs \(\sshift\subseteq\calG\),
which must be closed under the operations used throughout the paper.

Finally, names only serve as a description tool; their particular
choice in a graph should be immaterial. We therefore consider renamings, and
will require considered subset of graphs, neighbourhood schemes and local rules to commute with them.

\begin{definition}[Renaming]\label{def:renaming}
A renaming is a function $R:\calV\rightarrow \calV$ such that
$R(t.u)=t.R(u)$ and whose restriction to positions $\sp{R(0.-)} : \X \to \X$ is a bijection of $\X$.
It is fully specified by this restriction, and acts on the graphs of
Def.~\ref{def:graphs} by renaming their vertices.
\end{definition}

Based on these operations, one can restrict evolutions to a subset of graphs, as long as appropriate closure properties preserve their semantics.

\begin{definition}[Closed subset of graphs]\label{def:closed}
A subset of graphs $\sshift\subseteq\calG$ is \emph{closed} if it respects the following three properties.
\begin{description}[leftmargin=!,labelwidth=3.4cm,itemsep=.25ex]
\item[renaming] If $G\in \sshift$ implies $R(G)\in \sshift$ where $R$ is a
renaming.
\item[inclusion] If $G\in \sshift$ implies
$\forall H\sqsubseteq G,\ H\in \sshift$.
\item[disjoint-union] If $G,H\in \sshift$ and
$\posdisjoint{\V_G}{\V_H}$ implies $G\sqcup H\in \sshift$.
\end{description}
\end{definition}

These are the standard hypotheses of~\cite{arrighi2026causalgraphrewriting}.
From now on we fix a closed subset of graphs \(\sshift\), relative to which
every statement below is made.

\subsection{Neighbourhood schemes and local rules}\label{sec:local-rules}

A neighbourhood scheme returns the positions considered close to a position
\(x\). These positions are not arbitrary: they form a \emph{cone} of \(x\),
much as an event in relativity can only signal within its future light-cone,
or as a process sends its result to the processes that depend on it.

\begin{definition}[Cones]\label{def:cones}
A \emph{cone} of $x$ is a graph $C\in\sshift$ such that
$x\in\sp{\Past(C)}$ and, for every $v\in\I_C$, there exists a directed path
from the vertex at position $x$ to $v$ that lies entirely in $\I_C$, a.k.a.\
\emph{accessibility}. We denote by $\xCone$ the set of cones of $x$, and by
$\xcone{}$ the set of all cones.
\end{definition}

\begin{definition}[(Extensive) neighbourhood scheme]\label{def:neighbourhood}
An (extensive) \emph{neighbourhood scheme} $\restri{}$ is a function which maps
any position $x\in\calX$ to a function
$\restri{x}:\graphvalidexplicit{\restri{}}{x}\to\mathcal P(\calX)$ on some
domain $\graphvalidexplicit{\restri{}}{x}\subseteq\sshift$ such that for any
$H\in\sshift$ and $G\in\graphvalidexplicit{\restri{}}{x}$ we have:
\begin{description}[leftmargin=!,labelwidth=2.8cm,itemsep=.25ex]
\item[cone:] $\restr{x}{G}\subseteq\sp{\I_G}$ and
$G_{\restri{x}}$ is in $\xCone$.
\item[stability:] $G_{\restri{x}}\sqsubseteq H$ implies
$H\in\graphvalid{\restri{}}{x}$.
\item[extensivity:] $G_{\restri{x}}\sqsubseteq H\sqsubseteq G$ implies
$G_{\restri{x}}=H_{\restri{x}}$.
\item[renaming-invariance:] for all renaming $R$, we have
\[
R(G)\in\graphvalidexplicit{\restri{}}{\sp{R(x)}}
\quad\text{and}\quad
\restri{\sp{R(x)}}(R(G))=\sp{R(\restri{x}(G))}.
\]
\end{description}
where we introduced $G_{\restri{x}}$ as a shorthand for
$G_{\restr{x}{G}}$. We denote
$\diskA{x}:=\{G_{\restri{x}}\mid
G\in\graphvalidexplicit{\restri{}}{x}\}$ the \emph{set of disks} of $x$.
\end{definition}

The disks \(\diskA{x}\) are the neighbourhoods that can be extracted from
graphs of \(\sshift\); they are the building blocks on which local rules are
defined. By extensivity, the neighbourhood scheme leaves them unchanged
(Rem.~\ref{rem:idempotence}).

\begin{remark}[Extensivity implies idempotence]\label{rem:idempotence}
For any disk $G_{\restri{x}}\in\diskA{x}$, we have
$(G_{\restri{x}})_{\restri{x}}=G_{\restri{x}}$. Indeed, stability places
$G_{\restri{x}}$ in $\graphvalid{\restri{}}{x}$, and extensivity applied to
$G_{\restri{x}}\sqsubseteq G_{\restri{x}}\sqsubseteq G$ then gives
$\restr{x}{G_{\restri{x}}}=\restr{x}{G}$.
\end{remark}

A local rule \(A_{(-)}\) acts at \(x\) by replacing the disk
\(G_{\restri{x}}\) with \(A_xG_{\restri{x}}\), leaving the rest of the graph
unchanged.

\begin{definition}[Local rule]\label{def:locality}
Given some neighbourhood scheme $\restri{}$, an
\emph{$\restri{}$-local operator} $A_{(-)}$ is an operator from
$\mathcal{X}\times\sshift$ to $\sshift$ such that $\forall x\in\X$ and
$\forall G\in\sshift$,
$A_{\sp{R(x)}}(R(G))=R(A_xG)$ for every renaming $R$, and
\[
A_xG:=
\begin{cases}
(A_xG_{\restri{x}})\sqcup G_{\crestri{x}}
&\textrm{if $G\in\graphvalid{\restri{}}{x}$},\\
G&\textrm{otherwise,}
\end{cases}
\]
where the action of $A_x$ on $G_{\restri{x}}$ must preserve the
border---i.e.\ $\B_{G_{\restri{x}}}=\B_{A_xG_{\restri{x}}}$.
\end{definition}

By idempotence, \(A_x\) is determined by its action on disks.

\section{Borders as missing information}\label{sec:strong}

We want to read border vertices as missing information, and nothing more: they record
edges whose continuation has not been explored, not the end of the graph.
To a large extend, this is already the case with the above definition.

But taking this reading seriously requires two adjustments of the framework.
The first concerns the subset of graphs: missing information should always be possible to fill in.

\begin{definition}[Border-completable subset of graphs]\label{def:completable}
  A subset of graphs $\sshift$ is \emph{border-completable} when every
  $G\in\sshift$ has a \emph{completion}, i.e.\ some $H\in\sshift$ such that
  $\V_G\subseteq \I_H$ and $G\sqsubseteq H$.
\end{definition}

From now on, $\sshift$ is assumed to be border-completable. The second adjustment
concerns neighbourhoods: whatever lies beyond the border vertices of a disk,
the neighbourhood should stay the same.

\begin{definition}[Border-agnostic neighborhood scheme]\label{def:strong-ext}
  A neighbourhood scheme $\restri{}$ is \emph{border-agnostic} if and only if for all
  $G\in\graphvalid{\restri{}}{x}$, for all $H\in\sshift$, we have that
  $G_{\restri{x}}\sqsubseteq H$ implies
  $H_{\restri{x}}=G_{\restri{x}}$.
\end{definition}

This condition looks close to extensivity of Def.~\ref{def:neighbourhood}, but extensivity only compares a
graph with its subgraphs, whereas border-agnosticism compares a disk with
every graph that contains it. The example below show that the
difference matters: an extensive neighbourhood may depend on information
added beyond its borders (Sec.~\ref{sec:added-info}), and may have disks
that are selected only when their borders are borders of the whole graph,
on which local rules may behave arbitrarily (Sec.~\ref{sec:structural}).

This example and the following one use the same border-completable closed subset of graphs already considered in~\cite{ArrighiIJUC}. Take ports
$\pi=\{l,r\}$, standing for the left and right directions, and let $L$ be the
set of infinite \emph{line} graphs: every edge goes from a port $:\!l$ to a
port $:\!r$ or from $:\!r$ to $:\!l$, so that the vertices form an infinite
line without borders. We let $\sshift$ be the smallest closed subset of
graphs containing $L$; it is border-completable, since the borders of its graphs
can always be completed into lines.

\subsection{Stability under added information}\label{sec:added-info}

A first issue with neighbourhoods that are not border-agnostic is that
completing a disk may change the neighbourhood.

\begin{example}[Rounded neighbourhood]\label{ex:rounded}
  Consider the following neighbourhood scheme. $\restri{x}$ attempts to select all vertices in the
  cone in front of $x$, but only keeps an even number of vertices (other than
  $x$) in each direction (left and right). For instance, in the graph $G$
  shown in Fig.~\ref{fig:rounded}, $\restri{x}$ selects only $x$, since there
  are not enough vertices on the right to form an even number larger than
  $0$. In contrast, on the larger graph $H$, $\restri{x}$ can select two
  vertices on the right. This neighbourhood is extensive but not
  border-agnostic: $G_{\restri{x}}\sqsubseteq H$, but
  $H_{\restri{x}}\ne G_{\restri{x}}$.
\end{example}

\begin{figure}[t]
  \centering
  \begin{minipage}[b]{.23\textwidth}\centering
  \includegraphics[width=\linewidth]{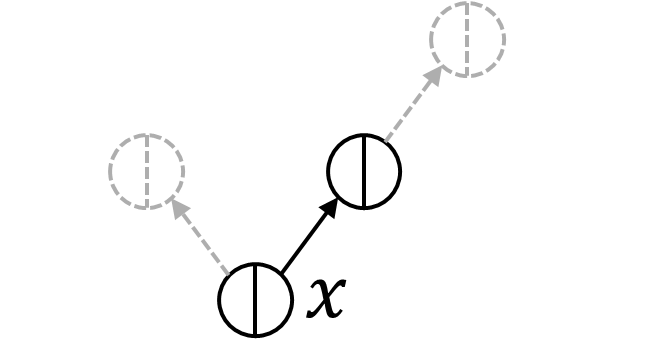}\\[-.4ex]\small \(G\)
  \end{minipage}\hfill
  \begin{minipage}[b]{.23\textwidth}\centering
  \includegraphics[width=\linewidth]{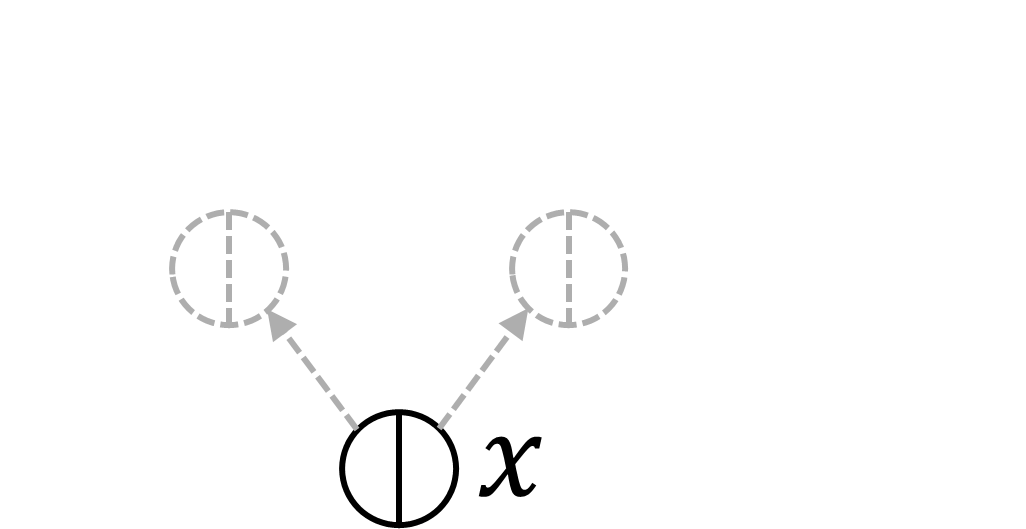}\\[-.4ex]\small
  \(G_{\restri{x}}\)
  \end{minipage}\hfill
  \begin{minipage}[b]{.23\textwidth}\centering
  \includegraphics[width=\linewidth]{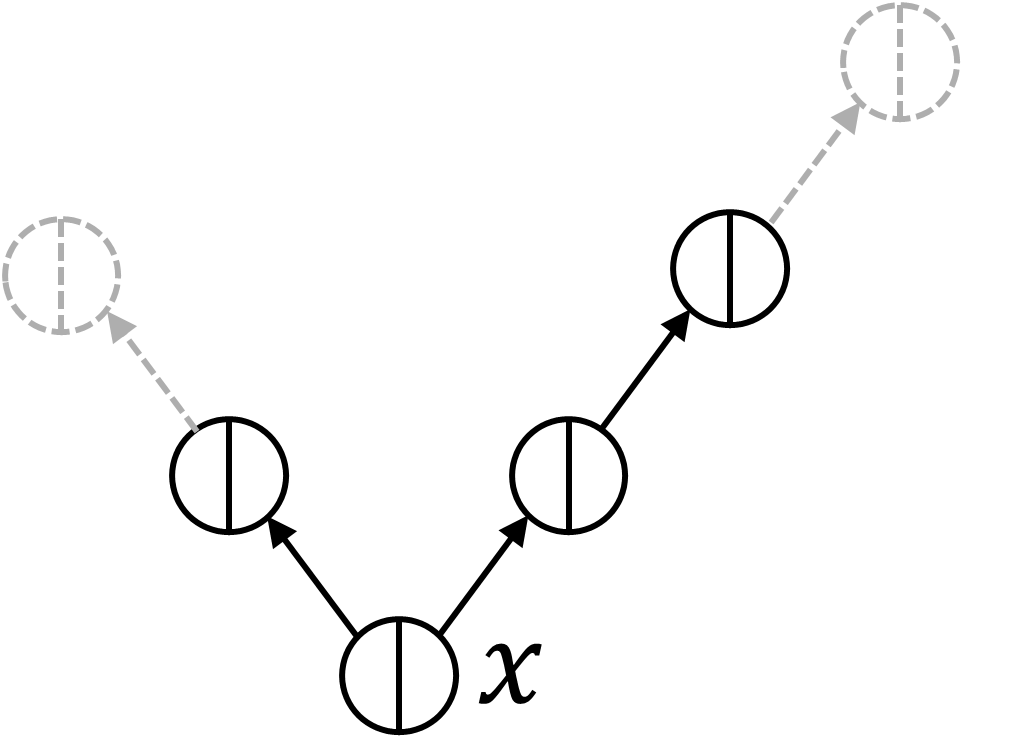}\\[-.4ex]\small \(H\)
  \end{minipage}\hfill
  \begin{minipage}[b]{.23\textwidth}\centering
  \includegraphics[width=\linewidth]{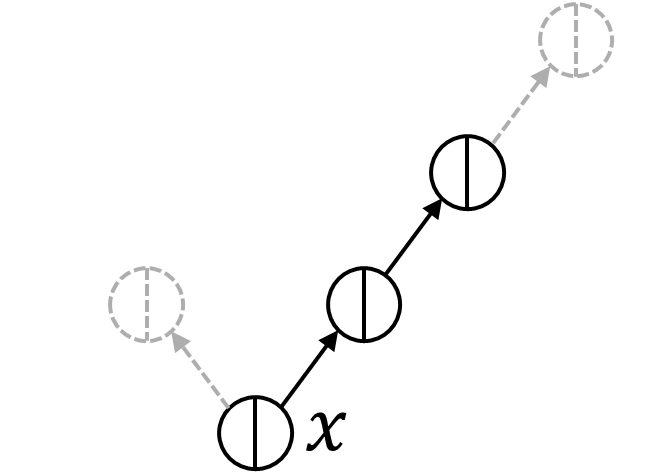}\\[-.4ex]\small
  \(H_{\restri{x}}\)
  \end{minipage}
  \caption{The rounded neighbourhood of Ex.~\ref{ex:rounded}, is 
  not border-agnostic. In $G$, $x$ has a single internal vertex on its right,
  which is therefore not selected: $G_{\restri{x}}$ only contains $x$. In
  $H$, which contains $G_{\restri{x}}$, there are two internal vertices on
  the right of $x$, and both are selected, whereas the single one on its left
  is not. Dashed vertices are borders. We do not represent ports here, but they allow us to distinguish between left and right neighbours.}
  \label{fig:rounded}
\end{figure}

In Ex.~\ref{ex:rounded}, completing $G_{\restri{x}}$ into $H$ changes the
neighbourhood: $H_{\restri{x}}$ selects more vertices, so the neighbourhood
computed on $G$ depended on information that was missing there.
Put differently, once $H_{\restri{x}}$ is known to be a neighbourhood, a
smaller graph $G$ that contains only part of it cannot know what its border
vertices hide, so the neighbourhood should remain undefined on $G$. We call
this requirement the \emph{safety principle}, and show that it is equivalent
to border-agnosticism.

\begin{definition}[Safety principle]\label{def:safety}
  An extensive neighbourhood scheme \(\restri{}\) obeys the \emph{safety
  principle} if, for all \(G\in\graphvalid{\restri{}}{x}\) and
  all \(H\sqsubseteq G\),
  \[
  G_{\restri{x}}\not\sqsubseteq H
  \quad\Longrightarrow\quad
  H\notin\graphvalid{\restri{}}{x}.
  \]
  That is, if the neighbourhood computed on a larger graph is not yet present
  in the smaller graph, then the smaller graph is not enough for computing the
  neighbourhood.
  \end{definition}
  
  \begin{proposition}[Border-agnosticism is equivalent to the safety principle]
  \label{prop:safety}
  An extensive neighbourhood scheme $\restri{}$ is border-agnostic if and
  only if it obeys the safety principle.
  \end{proposition}
  \begin{proof}
  \emph{Border-agnosticism implies the safety principle.}
  We prove the contrapositive. Suppose
  $H\in\graphvalid{\restri{}}{x}$ and $H\sqsubseteq G$. Then
  border-agnosticism gives
  $H_{\restri{x}}=G_{\restri{x}}$, which implies
  $G_{\restri{x}}\sqsubseteq H$.
  
  \emph{The safety principle implies border-agnosticism.}
  Let
  $G\in\graphvalid{\restri{}}{x}$ and $H\in\sshift$ such that
  $G_{\restri{x}}\sqsubseteq H$. Stability makes $H_{\restri{x}}$ defined.
  Apply the contrapositive of safety to the inclusion
  $G_{\restri{x}}\sqsubseteq H$. Since $G_{\restri{x}}$ is in the domain by
  stability, it follows that
  $H_{\restri{x}}\sqsubseteq G_{\restri{x}}$. Hence
  $H_{\restri{x}}\sqsubseteq G_{\restri{x}}\sqsubseteq H$, and extensivity
  applied to $H$ gives
  $H_{\restri{x}}=(G_{\restri{x}})_{\restri{x}}$. Idempotence concludes that
  $H_{\restri{x}}=G_{\restri{x}}$.
  \qed
  \end{proof}

The rounded neighbourhood violates precisely this principle: its stopping
decision treats a border as a wall, although the border only denotes missing
information.

\subsection{Structural consequences}\label{sec:structural}

The previous section shows how to translate the desire to treat borders as missing information into formal contraints guaranteeing stability under addition of information.
The first constraints is on closed subset of graphs that need to be border-completable, and the second one is on neighborhood schemes that need to be border-agnostic.
Under these constraints, the framework acquire more structure around the notion of disks.
The first two properties that we discuss now are disk uniqueness and disk-completability.
Together, they ultimately lead to the study, in the next sections, of local rules from the point of view of rewrite systems.

\paragraph{Disk uniqueness} states that a graph contains at most one disk at each position, which is readily crucial to the notion of matching a left hand side graph as required by rewrite systems.
More interestingly, under border-completion, disk uniqueness is actually equivalent to border-agnosticism.

\begin{lemma}\label{rem:disk-unique}
Border-agnosticism is equivalent to disk uniqueness, i.e. for any border-agnostic neighbourhood scheme $\restri{}$ and graph
$G\in\graphvalid{\restri{}}{x}$, there exists at most one disk
$D\in\diskA{x}$ such that $D\sqsubseteq G$.
\end{lemma}
\begin{proof}
For the first direction, suppose border-agnosticism of the neighborhood scheme, and consider two disks
$D^1$ and $D^2$ contained in $G$.
By idempotence we have $D^i_{\restri{x}}=D^i$.
Then border-agnosticism gives $G_{\restri{x}}=D^i$, which proves $D^1=D^2$.

In the other direction, disk uniqueness and the premise of the border-agnonisticism property, namely $G_{\restri{x}} \sqsubseteq H$.
By stability, $\restri{x}(H)$ is defined.
Since $H_{\restri{x}} \sqsubseteq H$, disk uniqueness concludes that $G_{\restri{x}} = H_{\restri{x}}$.
\end{proof}

As an illustration, the reader can check that Ex.~\ref{ex:rounded}, which is not border agnostic, do not have disk uniqueness property.
Indeed, in Fig.~\ref{fig:rounded}, $G_{\restri{x}}$ and $H_{\restri{x}}$ are two different disks included in $H$.

\paragraph{Disk-completability}

In the same way that border-completability asks for every graph to be induced from a bigger graph in which it is completely internal, asking for border-agnocism implies that every disk is a neighborhood at some position of a bigger graph in which it is completely internal.

\begin{lemma}[Border-agnosticism implies disk-completability]\label{lem:disk-completable}
Let \(\restri{}\) be a border-agnostic neighbourhood scheme.
For every disk \(D\in\diskA{x}\) there exists
\(H\in\sshift\) such that \(\V_D\subseteq\I_H\) and \(D=H_{\restri{x}}\).
\end{lemma}
\begin{proof}
Closure under inclusion of $\sshift$ gives \(D\in\sshift\), and
border-completability gives a completion \(H\in\sshift\) of \(D\), i.e.\
\(D\sqsubseteq H\) and \(\V_D\subseteq\I_H\). Idempotence gives
\(D_{\restri{x}}=D\), and border-agnosticism gives \(H_{\restri{x}}=D\).
\end{proof}

This implication is not an equivalence.
Indeed, Ex.~\ref{ex:rounded} is not border-agnostic but is disk-completable.
In the same vain, the following example shows, however, that not all neighborhood schemes are disk-completable.

This second kind of behaviour only shows up at
\emph{global borders}, i.e.\ borders of the whole graph under consideration,
beyond which nothing is known. Such borders can always be completed, but
some disks are selected only as long as their borders are global: in any
completion, the neighbourhood selects something else. If borders only denote
missing information, these disks are artefacts, and local rules should not
be able to behave specially on them. We illustrate this on the way a local
rule acting on a disk handles its border edges, which causal graph rewriting
does not force to be preserved.

\begin{example}[Neighbours if possible]\label{ex:neighbours}
  Consider the neighbourhood scheme where $\restri{x}$ selects $x$ together
  with those of its two neighbours that are internal vertices of the graph.
  In the graph $G$ of Fig.~\ref{fig:neighbours}, both neighbours of $x$ are
  border vertices, so $\restri{x}$ selects only $x$ and $G_{\restri{x}}=G$.
  In the completion $H$ of $G$, both neighbours are internal, and
  $\restri{x}$ selects them as well.
  This neighbourhood is extensive: if
  $G_{\restri{x}}\sqsubseteq H'\sqsubseteq G$, the selected neighbours are
  still internal in $H'$, and the others are still border vertices of $H'$,
  so $H'_{\restri{x}}=G_{\restri{x}}$.
  It is not disk-completable: in any completion $H'$ of $G$, both neighbours
  of $x$ are internal, so $H'_{\restri{x}}$ contains them and differs from
  $G_{\restri{x}}$. By Lem.~\ref{lem:disk-completable}, it is therefore not
  border-agnostic. Other disks are disk-completable, such as
  $H_{\restri{x}}$, whose border vertices are internal in $H$.
\end{example}

\begin{figure}[t]
  \centering
  \begin{minipage}[b]{.3\textwidth}\centering
  \includegraphics[width=\linewidth]{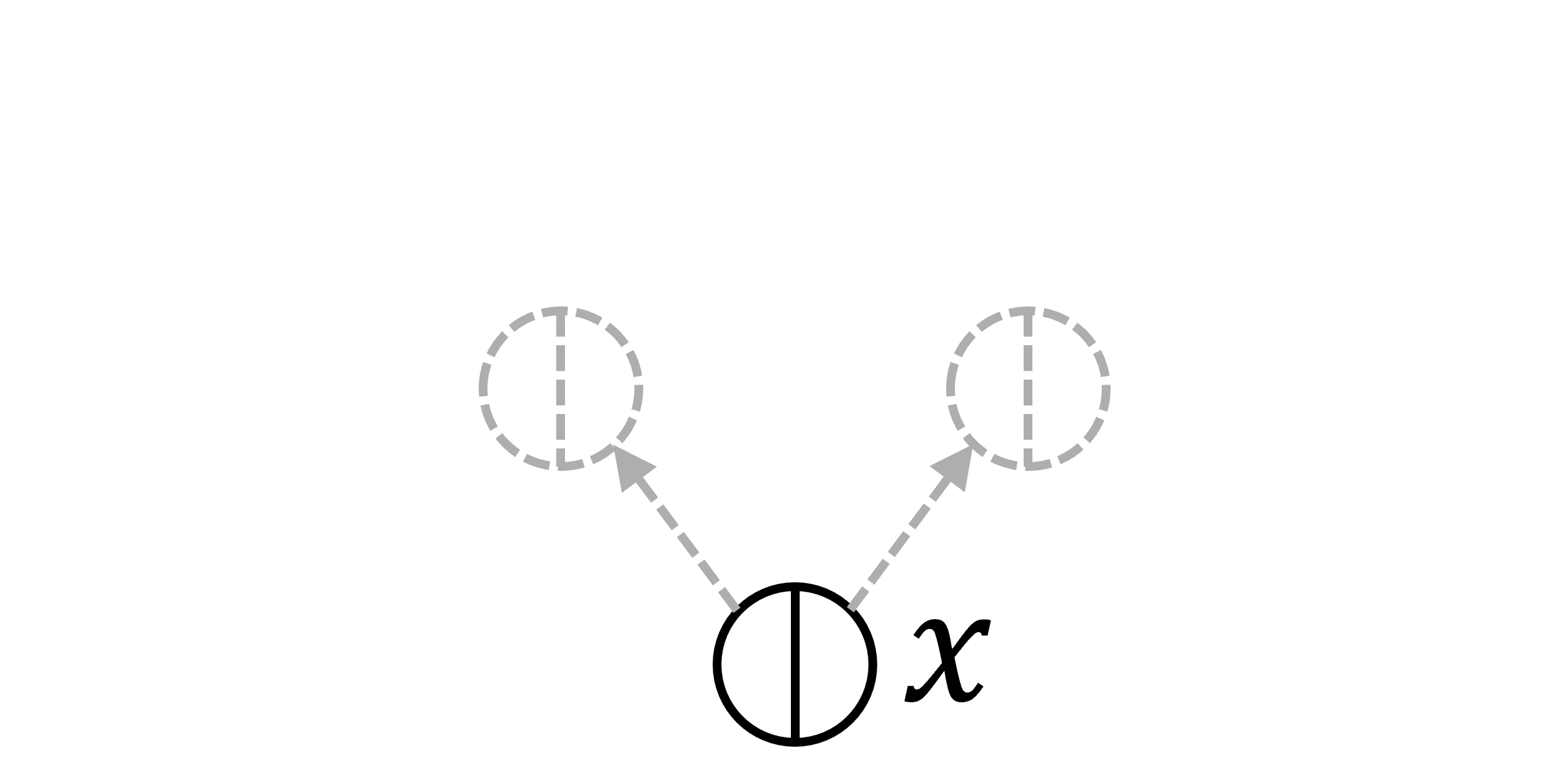}\\[-.4ex]\small
  \(G=G_{\restri{x}}\)
  \end{minipage}\hfill
  \begin{minipage}[b]{.3\textwidth}\centering
  \includegraphics[width=\linewidth]{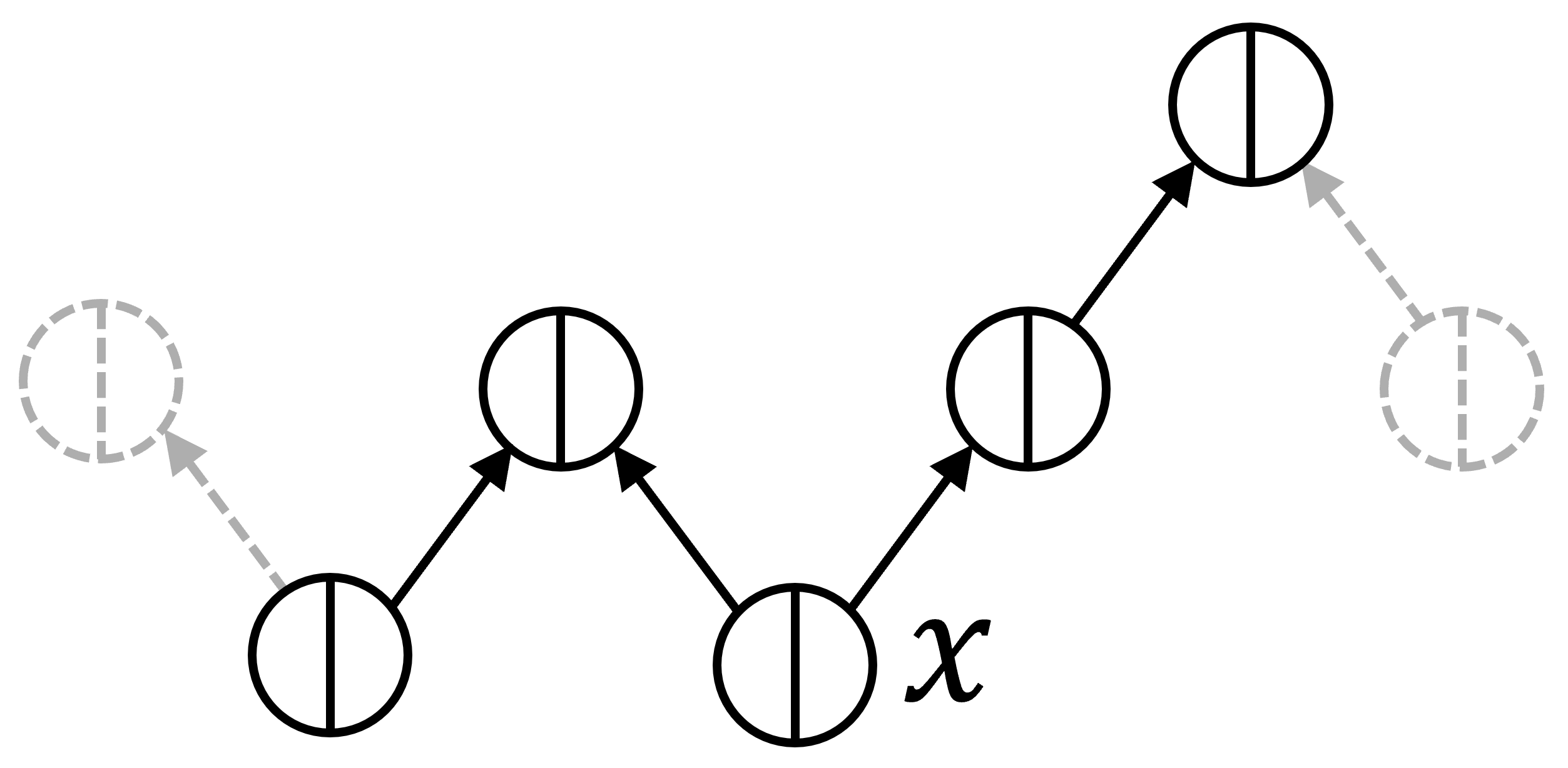}\\[-.4ex]\small \(H\)
  \end{minipage}\hfill
  \begin{minipage}[b]{.3\textwidth}\centering
  \includegraphics[width=\linewidth]{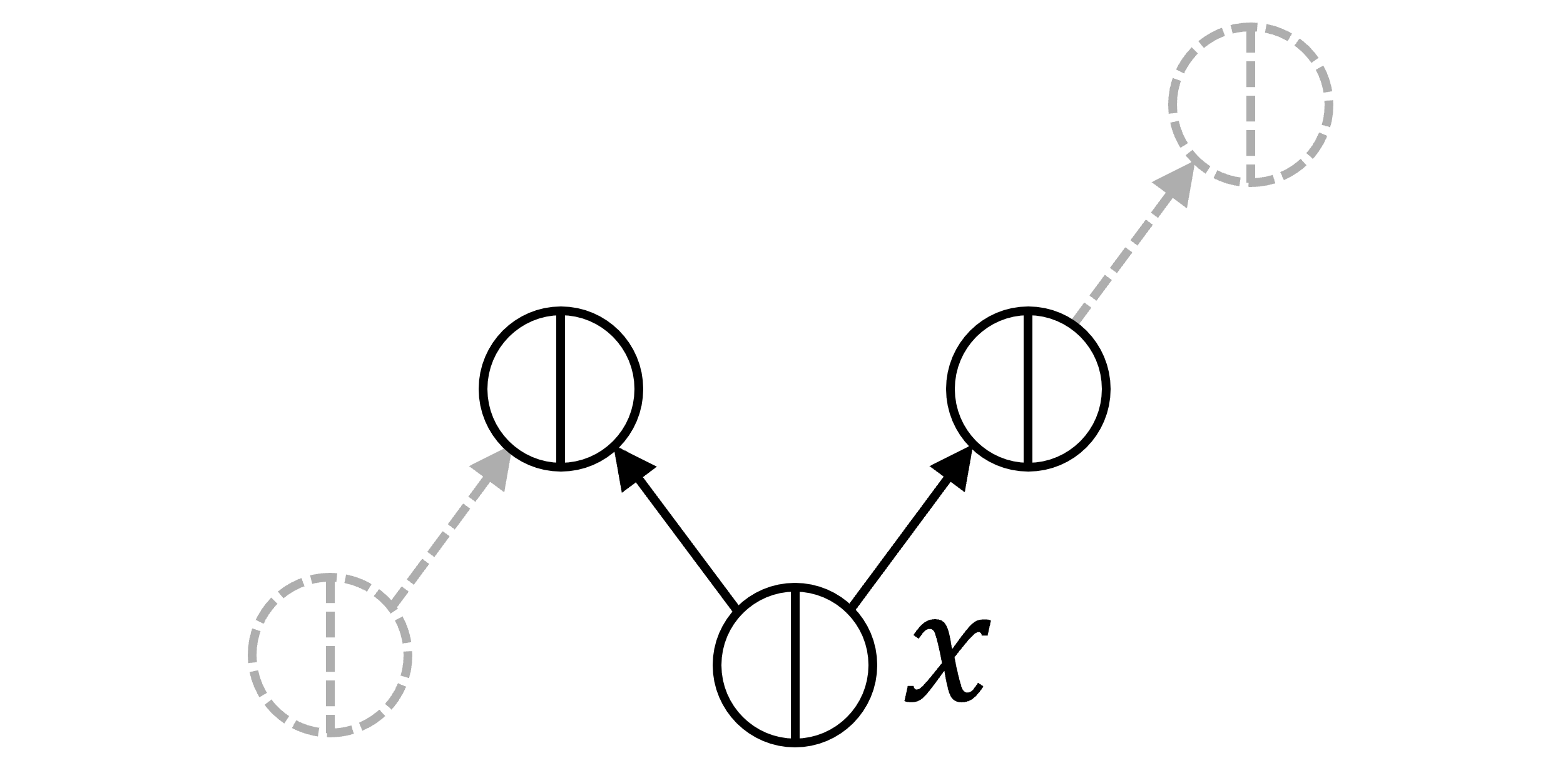}\\[-.4ex]\small
  \(H_{\restri{x}}\)
  \end{minipage}
  \caption{Ex.~\ref{ex:neighbours}: the graph $G=G_{\restri{x}}$ (left), a
  completion $H$ of $G$ (middle) and its disk $H_{\restri{x}}$ (right). The
  disk $H_{\restri{x}}$ is disk-completable: its border vertices are
  internal in $H$. The disk $G_{\restri{x}}$ is not: in any completion of
  $G$, both neighbours of $x$ become internal and are selected.}
  \label{fig:neighbours}
\end{figure}

\section{From neighbourhoods to rewrite patterns}\label{sec:crs}

\subsection{Mutex sets of cones}

Our goal is to present local rules as rewrite systems. Their left-hand sides
are what the neighbourhood scheme selects, that is, the disks. We therefore
characterise the families of graphs $\D_x$ that arise as the disks
$\diskA{x}$ of some border-agnostic neighbourhood scheme. Three conditions are
needed. First, $\D_x$ contains only cones of $x$, reflecting the cone
condition. Second, a graph contains at most one element of $\D_x$: this
\emph{unambiguity} condition reflects border-agnosticism, which we proved
equivalent to having at most one disk at $x$ in any graph
(Rem.~\ref{rem:disk-unique}). Third, the family is invariant under renaming,
reflecting renaming-invariance.

\begin{definition}[Mutex]\label{def:mutex}
A map {$\mathcal{D}:\X\to\mathcal{P}(\xcone{}),\ x\mapsto \mathcal{D}_x$} defines \emph{a mutex set of cones} iff it is
\begin{description}[leftmargin=!,labelwidth=2.8cm,itemsep=.25ex]
\item[\RX{conic:}{well-indexed:}] {$\mathcal{D}_x\subseteq\xcone{x}$}\RX{}{ and $\mathcal{D}_x\cap \mathcal{D}_y\neq \varnothing$ implies $x=y$}.
\item[unambiguous:] $\forall x\in\X,\ \forall G \in \sshift,\
|\{ D\in\mathcal{D}_x \mid D\sqsubseteq G\}|\leq 1$.
\item[renaming-invariant:] for every renaming $R$,
$\mathcal{D}_{\sp{R(x)}}=\{RD\mid D\in\mathcal{D}_x\}$.
\end{description}
\end{definition}

The disks of a neighbourhood scheme indeed form a mutex set of cones.

\smallskip
\begin{lemma}[From neighbourhoods to mutex cones]\label{lem:neigh-to-mutex}
Let $\restri{}$ be a border-agnostic neighbourhood scheme. Then the function $\diskA{} : \X\to\calP(\xcone{}),\ x\mapsto \diskA{x}$ is a mutex set of cones.
\end{lemma}
\begin{proof}
By the cone condition, each $G_{\restri{x}}$ lies in $\xcone{x}$, so $\diskA{}$ is \RX{conic}
{well-indexed}. Renaming invariance follows from that of $\restriA{}$ and from the closure of $\sshift$ under renaming.
Unambiguity is Rem.~\ref{rem:disk-unique}, since by stability any graph of
$\sshift$ containing a disk at $x$ lies in $\graphvalid{\restri{}}{x}$.
\end{proof}

\smallskip
Conversely, a mutex set of cones \(\D\) defines a neighbourhood scheme: in a
graph \(G\), the neighbourhood of \(x\) is the element of \(\D_x\) contained
in \(G\), when there is one.

\begin{definition}[Induced neighbourhood scheme]\label{def:induced-neigh}\label{def : induced neighbourhood scheme}
We call
\[
\graphvalidmutex{\D}{x}
=\{G\in\sshift\mid\exists D\in\D_x,\ D\sqsubseteq G\}
\]
the set of graphs containing some $D\in\mathcal{D}_x$. Given a graph
$G\in\graphvalidmutex{\D}{x}$, unambiguity makes this $D$ unique.
This allows us to define
\[
\restriD{\D}{}:x\mapsto\restriD{\D}{x},
\qquad
\restriD{\D}{x}:\graphvalidmutex{\D}{x}\to\mathcal P(\X),
\quad
\restrD{\D}{x}{G}=\sp{\I_D}.
\]
\end{definition}

The induced scheme is indeed a neighbourhood scheme.

\begin{lemma}[From mutex cones to neighbourhoods]\label{lem:mutex-to-neigh}
Let \(\D\) be a mutex set of cones. Then \(\restriD{\D}{}\) is a well-defined
border-agnostic neighbourhood scheme.
\end{lemma}
\begin{proof}
\emph{Renaming-invariance.}
Suppose \(D\in\D_x\) and \(D\sqsubseteq G\). Renaming invariance of \(\D\)
and closure of \(\sshift\) give
\(RD\in\D_{\sp{R(x)}}\) and \(RD\sqsubseteq RG\). Unambiguity makes these
selected patterns unique, hence
\[
RG_{\restriD{\D}{x}}
=RD
=(RG)_{\restriD{\D}{\sp{R(x)}}}.
\]
\emph{Cone.}
The selected \(D\) lies in \(\xcone{x}\) by the cone condition of mutex.
\emph{Border-agnosticism.}
If \(G_{\restriD{\D}{x}}=D\sqsubseteq H\), then
\(H\in\graphvalidmutex{\D}{x}\); unambiguity makes \(D\) its unique selected
pattern. Thus \(H_{\restriD{\D}{x}}=D\).
\end{proof}

The two constructions are inverse of one another.

\begin{proposition}[Equivalence between neighbourhood schemes and mutex sets of cones]\label{prop:neigh-mutex}
The function mapping any border-agnostic neighbourhood scheme $\restri{}$ to $\diskA{}$ (see Lem.~\ref{lem:neigh-to-mutex}) is a bijection between border-agnostic neighbourhood schemes on the one hand, and mutex sets of cones on the other. The inverse associates, to any mutex set of cones $\D$, the neighbourhood scheme $\restriD{\D}{}$ (see Lem.~\ref{lem:mutex-to-neigh}).
\end{proposition}
\begin{proof}
\emph{Right inverse.} Let $\restriB{}$ be a neighbourhood scheme, we want $\restriB{} = \restriD{\sshift_{\restriB{}}}{}$. Let $x\in \mathcal{X}$ be a position. First we prove that $\graphvalid{\restriB{}}{x}$ the domain of $\restriB{x}$ is equal to $\graphvalidmutex{\diskB{}}{x}$ the domain of $\restriD{\diskB{}}{x}$. Let us unravel the definitions. On one side $\graphvalidmutex{\diskB{}}{x}$ is the set of all graph $G$ such that there exists $D\in \diskB{x}$ such that $D\sqsubseteq G$. For all $G\in \graphvalid{\restriB{}}{x}$ we have $G_{\restri{x}}\sqsubseteq G$ and so $\graphvalid{\restriB{}}{x} \subseteq \graphvalidmutex{\diskB{}}{x}$. Next, by border-agnosticism $\restriB{}$ is defined on all graphs containing a disk, so that we also have $ \graphvalidmutex{\diskB{}}{x}\subseteq \graphvalid{\restriB{}}{x}$.\\
Then we consider $G\in \graphvalid{\restriB{}}{x}=\graphvalidmutex{\diskB{}}{x}$. By definition there exists $D \in \diskB{x}$ such that $D \sqsubseteq G$. We then have, by border-agnosticism, $G_{\restriD{\diskB{}}{x}}=D=G_{\restriB{x}}$ from which $\restriB{} = \restriD{\sshift_{\restriB{}}}{}$ follows.

\emph{Left inverse.}
Now we consider {$\D:x\mapsto\D_x$} a mutex set of cones, and we prove that ${\sshift}_{\restriD{\D}{}} = \D$. First we check that $\D_x\subseteq {\sshift}_{\restriD{\D}{x}}$. By definition of $\restriD{\D}{x}$, we have $D_{\restriD{\D}{x}} = D$ which implies $D\in {\sshift}_{\restriD{\D}{x}}$. The reverse inclusion ${\sshift}_{\restriD{\D}{x}}\subseteq \D_x$ is by definition.
\end{proof}

\subsection{Rewrite system}

By idempotence (Rem.~\ref{rem:idempotence}), a local rule \(A_x\) is
determined by the rules \(\{D\to A_xD\}_{D\in\diskA{x}}\). Conversely, we
want to recover a local rule from such rules alone. But a local rule is defined as
a global function on \(\sshift\), so the cup
\(A_xG_{\restri{x}}\sqcup G_{\crestri{x}}\) is well defined by assumption;
a rewrite system only provides the rules, and must instead guarantee that
replacing a left-hand side by its right-hand side within any graph yields a
well-defined graph. This is the role of context preservation, a local
condition checked on each rule alone: it makes the cup of the right-hand
side with any context of the left-hand side well defined, up to acyclicity
and finite chains, which are global and required separately. Functionality
and a mutex domain (Prop.~\ref{prop:neigh-mutex}) guarantee that at most one
update applies at each position; renaming-invariance makes the rules
independent of the choice of names; and \(\sshift\)-preservation, keeps the result in \(\sshift\).

\begin{definition}[Causal rewrite system]\label{def:crs}
A set of rewrite rules
$\{D_x^j\to G_x^j\}_{x\in\X,j\in J}$ specifies a
\emph{causal rewrite system} over $\sshift$ iff it is
\begin{description}[style=unboxed,leftmargin=0pt,itemindent=0pt,labelsep=.45em,itemsep=.35ex,parsep=0pt]
\item[functional:] $D_x^j=D_y^k$ implies $(x,j)=(y,k)$, so that
$f:D_x^j\mapsto G_x^j$ is a function;
\item[renaming-invariant:] $f(RD_x^j)=R(G_x^j)$ for every renaming $R$;
\item[mutex-domain:] $\mathcal D:x\mapsto\{D_x^j\mid j\in J\}$ defines a
mutex set of cones;
\item[context-preserving:]
$\possubseteq{\I_{G_x^j}}{\I_{D_x^j}}$,
$\B_{G_x^j}=\B_{D_x^j}$, and
$\dangling{G_x^j}=\dangling{D_x^j}$;
\item[acyclicity and finite chains:] whenever \(D_x^j\sqcup H\) is well defined
and lies in \(\sshift\), the union of the edge sets of \(G_x^j\) and \(H\) is
acyclic and contains only finite chains;
\item[$\sshift$-preserving:] whenever \(D_x^j\sqcup H\) and \(G_x^j\sqcup H\) are well defined,
\[
D_x^j\sqcup H\in\sshift\implies G_x^j\sqcup H\in\sshift.
\]
\end{description}
\end{definition}

To show that the rules \(\{D\to A_xD\}\) form a CRS, the main point is context
preservation, which follows from the structural properties of local rules
established in~\cite{arrighi2026causalgraphrewriting}. Border vertices are
preserved by Def.~\ref{def:locality}, and positions
by~\cite[Lemma~3.26]{arrighi2026causalgraphrewriting}. Border edges require
more: \cite[Lemma~3.28]{arrighi2026causalgraphrewriting} only preserves them
on disk-completable disks, and we saw in
Sec.~\ref{sec:structural} that some disks with global borders are not.
Lem.~\ref{lem:disk-completable} rules out such disks, so the lemma applies to
every disk.
\ifnamealgebra
With the name algebra, the argument of that lemma carries over except for one inclusion. Preservation of boundary vertices, and
\(\dangling{D}\subseteq\dangling{A_xD}\), do not use position preservation,
so both carry over unchanged. The reverse inclusion in that proof uses
\(\sp{\I_{A_xD}}\subseteq\sp{\I_D}\) to conclude that the internal endpoint
was already a vertex of \(\I_D\). Algebraic position preservation only places
that endpoint in \(\namealg{\I_D}\). The inclusion still holds: the endpoint
is a vertex of the completion, it shares a position fragment with a vertex
of \(\I_D\), and non-overlapping positions force the two to coincide.
\fi

\ifnamealgebra

\fi

\begin{lemma}[From local rules to CRS]\label{lem:local-to-crs}
For any border-agnostic neighbourhood scheme \(\restri{}\) and
\(\restri{}\)-local rule \(A_{(-)}\), the family of rules
\[
\{D\to A_xD\}_{D\in\diskA{x},\,x\in\X}
\]
is a causal rewrite system.
\end{lemma}
\begin{proof}
It is functional by definition. Prop.~\ref{prop:neigh-mutex} gives the
mutex domain, and renaming invariance comes from \(A_{(-)}\).
Def.~\ref{def:locality} preserves border vertices.
\ifnamealgebra
Lem.~\ref{lem:pos-preservation} gives
\(\I_{A_xD}\namesubseteq\I_D\),
\else
\cite[Lemma~3.26]{arrighi2026causalgraphrewriting} gives
\(\sp{\I_{A_xD}}\subseteq\sp{\I_D}\),
\fi
and border edges are preserved: Lem.~\ref{lem:disk-completable} provides
\(H\in\sshift\) with \(\V_D\subseteq\I_H\) and \(H_{\restri{x}}=D\), and
applying~\cite[Lemma~3.28]{arrighi2026causalgraphrewriting} to \(H\) gives
\(\dangling{A_xD}=\dangling{D}\).
Finally, whenever \(D\sqcup H\) is well defined and lies in \(\sshift\),
border-agnosticism selects \(D\) as the neighbourhood of \(D\sqcup H\), so the
locality equation reads \(A_x(D\sqcup H)=A_xD\sqcup H\). The left-hand side
is a graph in \(\sshift\), hence \(A_xD\sqcup H\) is defined, acyclic, of
finite chains, and in \(\sshift\) because \(A_x\) sends \(\sshift\) into \(\sshift\).
This proves acyclicity, finite chains, and \(\sshift\)-preservation.
\end{proof}

Conversely, a CRS acts at \(x\) on a graph \(G\) by replacing the
left-hand side \(D_x^j\) contained in \(G\), which is unique by
unambiguity, with \(G_x^j\); when \(G\) contains no
left-hand side at \(x\), it is left unchanged.

\begin{definition}[Rewriting a graph using a CRS]\label{def:crs-action}
Consider a causal rewrite system
\(\calR=\{D_x^j\to G_x^j\}_{x\in\X,j\in J}\). We define
\(A^\calR_{(-)}:x\mapsto A_x^\calR\), with $\forall G\in 
\sshift$,
\[
A_x^\calR G=
\begin{cases}
G_x^j\sqcup G_{\overline{\sp{\I_{D_x^j}}}}
&\text{if }\exists j\in J,\
G=D_x^j\sqcup G_{\overline{\sp{\I_{D_x^j}}}},\\
G&\text{otherwise.}
\end{cases}
\]
\end{definition}

The conditions of Def.~\ref{def:crs} are exactly what is needed
for this replacement to be a well-defined graph of \(\sshift\), so that the
action is a local rule.

\begin{lemma}[From CRS to local rules]\label{lem:crs-to-local}
For any causal rewrite system
$\calR=\{D_x^j\to G_x^j\}_{x\in\X,j\in J}$, the function
\(A^\calR_{(-)}\) is a \(\restriD{\D}{}\)-local rule, where
\(\D:x\mapsto\{D_x^j\mid j\in J\}\).
\end{lemma}
\begin{proof}
We first show that Def.~\ref{def:crs-action} is well defined.
If \(G\in\sshift\) contains no disk of \(\D\) at \(x\), the definition
returns \(G\).
Otherwise, functionality and the mutex condition give a unique image
\(G_x^j\).
Let \(C=G_{\overline{\sp{\I_{D_x^j}}}}\), so \(G=D_x^j\sqcup C\) and the
definition returns \(G_x^j\sqcup C\).
It remains to prove that this cup is well defined. 
Context preservation gives \(\possubseteq{\I_{G_x^j}}{\I_{D_x^j}}\), and the 
decomposition of \(G\) gives \(\I_{D_x^j}\cap\I_C=\varnothing\), so the
internal vertices of \(G_x^j\) and of \(C\) are disjoint.
\cite[Lemma~3.10]{arrighi2026causalgraphrewriting} then characterises
well-definedness, with its position clause read as
\ifnamealgebra non-overlapping of position fragments\else distinct
positions\fi\ (Def.~\ref{def:graphs}).
\ifnamealgebra
Non-overlapping holds by the second claim of Lem.~\ref{lem:fragments}: a
vertex \(w\) of \(G_x^j\) lies in \(\namealg{\I_{D_x^j}}\), so if \(w\)
overlapped a vertex \(c\) of \(C\), then some vertex of \(\I_{D_x^j}\) would
overlap \(c\) as well, which the non-overlapping condition of the graph
\(D_x^j\sqcup C\) forbids.
\else
Distinct positions hold as well: by context preservation, every position of
\(G_x^j\) is a position of \(D_x^j\), so a position shared with \(C\) is
carried in \(D_x^j\sqcup C\) by a vertex of both \(D_x^j\) and \(C\). Either
it is a border vertex of \(D_x^j\), kept by \(G_x^j\) since borders are
preserved; or it is an internal vertex of \(D_x^j\) on a border edge, kept by
\(G_x^j\) since border edges are preserved. In both cases \(G_x^j\) and \(C\)
carry the same vertex at that position.
\fi
Dangling coherence and unsaturated ports are the rest of context
preservation.
Border vertices are unchanged and
\(\dangling{G_x^j}=\dangling{D_x^j}\), so a border edge of \(G_x^j\) that
meets an internal vertex of \(C\) is already a border edge of \(D_x^j\),
hence lies in \(C\), and symmetrically.
The ports used on those common border vertices are the ports of
\(D_x^j\sqcup C\).
Acyclicity and finite chains are the remaining clause of \cite[Lemma~3.10]{arrighi2026causalgraphrewriting}.
They are the corresponding hypothesis of Def.~\ref{def:crs}: since
\(D_x^j\sqcup C\in\sshift\), the union of the edge sets of \(G_x^j\) and \(C\)
is acyclic and contains only finite chains.
The four clauses hold, so \(G_x^j\sqcup C\) is well defined.\\
Finally, \(A^\calR_{(-)}\) is a local rule.
Lem.~\ref{lem:mutex-to-neigh} makes \(\restriD{\D}{}\) a border-agnostic
neighbourhood scheme, and the first case of
Def.~\ref{def:crs-action} is the locality equation for the disk it selects;
outside that domain the action is the identity.
The signature \(\sshift\to\sshift\) follows: on a graph with no disk the
value is the graph itself, and once \(G_x^j\sqcup C\) is well defined,
\(\sshift\)-preservation yields \(G_x^j\sqcup C\in\sshift\).
Renaming invariance is inherited from the rule family, since \(\sshift\) is
closed under renaming and \(\restriD{\D}{}\) is renaming-invariant.
Context preservation gives \(\B_{G_x^j}=\B_{D_x^j}\), so the action
preserves the border.
\end{proof}

Finally, the constructions of Lem.~\ref{lem:local-to-crs} and
Def.~\ref{def:crs-action} are inverse of one another. This establishes a
formal correspondence between the local rules of causal graph rewriting and
rewrite systems in the traditional sense of pattern replacement.

\begin{theorem}[Equivalence between local operators and causal rewrite
systems]\label{thm:representation}
The function mapping any pair of a border-agnostic neighbourhood scheme
\(\restri{}\) and a \(\restri{}\)-local rule \(A_{(-)}\), to the causal rewrite
system
\[
\{D\to A_xD\}_{D\in\diskA{x},\,x\in\X}
\]
is a bijection, whose inverse associates to any causal rewrite system
\(\calR=\{D_x^j\to G_x^j\}_{x\in\X,j\in J}\) the pair formed by
\(\restriD{\D}{}\), with
\(\D:x\mapsto\{D_x^j\mid j\in J\}\), and the
\(\restriD{\D}{}\)-local operator \(A^\calR_{(-)}\).
\end{theorem}
\begin{proof}
The two constructions are well defined by
Lems.~\ref{lem:local-to-crs} and~\ref{lem:crs-to-local}.
Starting with \((\restri{},A)\),
Prop.~\ref{prop:neigh-mutex} gives
\(\restriD{\diskA{}}{}=\restri{}\). For all \(x\in\X\) and
\(G=D_x^j\sqcup G_{\crestri{x}}\),
\begin{align*}
A_x^\calR G
&= G_x^j\sqcup G_{\crestriD{}{x}}
= A_x D_x^j\sqcup G_{\overline{\restriD{}{x}}}\\
&= A_x G_{\restriD{}{x}}\sqcup G_{\overline{\restriD{}{x}}}
= A_x G.
\end{align*}
Starting with \(\calR\), its induced local rule satisfies
\(A_x^\calR D_x^j=G_x^j\), while the disks of
\(\restriD{\D}{}\) are exactly the elements of \(\D\). Reconstructing the
rule family therefore returns exactly \(\calR\).
\end{proof}

\section{Related work and discussion}\label{sec:related}

The graph-dynamics lineage motivating this work starts from causal graph
dynamics~\cite{ArrighiCGD}, including reversible and quantum
variants~\cite{ArrighiRCGD,ArrighiQCGD}. Quantum Network
Theory~\cite{ArrighiQNT} isolates extensibility and idempotence as conditions
under which neighbourhood restrictions can be promoted to the quantum
setting. Space-time deterministic graph
rewriting~\cite{arrighi2024spacetimedeterministicgraphrewriting} replaces
synchronous global evolution by asynchronous local updates on dependency
DAGs. The present paper does not change that operational model; it strengthens the conditions on neighbourhoods on
which they are built.

Traditional graph rewriting starts from syntactic rules and formalises their
application by constructions such as double-pushout rewriting; modern
compositional accounts make interfaces and contexts
explicit~\cite{HarmerFundamentals}. Our direction is converse: we start from
the axiomatic local rules used for graph dynamics and identify when they
admit a pattern-replacement presentation. Causal rewrite systems are a first
step towards a double-pushout characterisation of causal graph rewriting:
their rules are given by explicit left- and right-hand sides, and context
preservation fixes the borders along which they are glued to the context,
much as the interface of a double-pushout rule.

\ifnamealgebra
The name algebra originates in work reconciling reversible locality with
vertex creation and destruction~\cite{ArrighiCreation}. Here time tags
commute with splitting and merging. Border-agnosticism is orthogonal to that
construction: the former controls ambient context, while the latter records
which apparent creations reuse algebraic material from the input.
\else
A first extension would be to replace the acyclicity and finite-chains
conditions of causal rewrite systems by local criteria, yielding a fully
constructive definition: which pattern substitutions are guaranteed to map
DAGs to DAGs? A second extension is reversibility.  The present
characterisation should be useful there. First, unique pattern matching
around each position, which is what border-agnosticism provides, seems a
natural requirement in a reversible setting. Second, the local conditions
defining causal rewrite systems could lead to local, constructive criteria
for reversibility.
\fi

\section{Summary and Conclusion}\label{sec:conclusion}

We studied neighbourhood schemes in causal graph rewriting, and introduced
border-agnosticism (Def.~\ref{def:strong-ext}): the neighbourhood selected
in a disk is the one selected in any graph of \(\sshift\) containing it.
Extensivity alone does not imply it (Ex.~\ref{ex:rounded}). We justified
border-agnosticism in three ways. First, it is equivalent to the safety
principle, which reads borders as missing information
(Prop.~\ref{prop:safety}). Second, it is equivalent to disk uniqueness: any
graph contains at most one disk at each position (Rem.~\ref{rem:disk-unique}).
Third, it rules out disks whose borders are global, on which local rules may
destroy border edges (Ex.~\ref{ex:neighbours}): on a border-completable subset of
graphs (Def.~\ref{def:completable}), every border-agnostic scheme is
disk-completable (
Lem.~\ref{lem:disk-completable}), so the
structural properties of~\cite{arrighi2026causalgraphrewriting} hold on
every disk.

Building on these results, we gave a representation of local rules by
pattern replacement. Border-agnostic neighbourhood schemes are in bijection
with mutex sets of cones (Def.~\ref{def:mutex}, Prop.~\ref{prop:neigh-mutex}),
the unambiguity condition being the counterpart of disk uniqueness. Local
rules are in bijection with causal rewrite systems (Def.~\ref{def:crs},
Thm.~\ref{thm:representation}): every local rule induces one
(Lem.~\ref{lem:local-to-crs}), and every causal rewrite system defines a
local rule (Lem.~\ref{lem:crs-to-local}). In the latter direction, context
preservation guarantees, up to acyclicity and finite chains, that replacing
a pattern in any graph is well defined.

Ordinary extensivity is thus sufficient for idempotence, but not for
interpreting borders as missing information, nor disks as
context-independent left-hand sides; border-agnosticism closes this gap.

\subsubsection*{Data Availability Statement}
No data are associated with this paper.

\bibliographystyle{splncs04}
\bibliography{biblio}
\end{document}